\documentclass{amsart}
\usepackage{mathrsfs,amssymb,amsmath,amsfonts,amsxtra}
\usepackage{verbatim}
\usepackage{graphicx,color}

\begin{document}
\title[Pilot wave beables and decoherence]{About the relation between pilot wave beables and decoherence}

\author{I. Schmelzer}
\thanks{Berlin, Germany}
\email{\href{mailto:ilja.schmelzer@gmail.com}{ilja.schmelzer@gmail.com}}%
\urladdr{\href{http://ilja-schmelzer.de}{ilja-schmelzer.de}}
\sloppypar \sloppy

\begin{abstract}
Motivated by Wallace's thesis that pilot wave beables should be decoherence-preferred to recover quantum predictions, we consider the relation between pilot wave beables and decoherence. 

We prove that without any connection between beables and decoherence the overlap between macrcopic states becomes negligible. This is sufficient to recover quantum predictions, so that Wallace's thesis has to be rejected.

A natural connection between decoherence and beables appears if the decomposition into systems used by decoherence is based on the beables. While our first result becomes inapplicable in this case, we present evidence that the overlap becomes negligible too.
\end{abstract}

\newcommand{\pd}{\partial} 
\renewcommand{\d}{\mathrm{d}} 
\renewcommand{\c}{\mbox{$\mathrm{const}$}}
\renewcommand{\Im}{\mbox{$\mathfrak{Im}$}}
\newcommand{\vol}[1]{\mbox{$A_{#1}$}}

\newcommand{\B}{\mbox{$\mathbb{Z}_2$}} 
\newcommand{\Z}{\mbox{$\mathbb{Z}$}}
\newcommand{\R}{\mbox{$\mathbb{R}$}}
\newcommand{\C}{\mbox{$\mathbb{C}$}}
\renewcommand{\H}{\mbox{$\mathcal{H}$}} 
\renewcommand{\L}{\mbox{$\mathcal{L}^2$}} 
\newcommand{\T}{\mbox{$\mathcal{T}$}} 
\newcommand{\V}{\mbox{$\mathcal{V}$}} 

\newtheorem{theorem}{Theorem}
\newtheorem{criterion}{Criterion}

\providecommand{\abs}[1]{\lvert#1\rvert}
\providecommand{\loo}[1]{\lVert#1\rVert_\infty}
\providecommand{\ltwo}[1]{\lVert#1\rVert_2}

\maketitle

\section{Introduction}

Reading in Wallace's paper \cite{Wallace} that
\begin{quote}
``The plain truth is that there are currently no hidden-variable \ldots theo\-ries which are generally accepted to reproduce the empirical predictions of any interacting quantum field theory. This is a separate matter to the conceptual problems with such strategies \ldots We do not even have QFT versions of these theories to have conceptual problems with.'' \cite{Wallace},
\end{quote}
I was very surprised. In my opinon, the problem of \emph{existence} of a realistic theory which allows to recover the \emph{empirical} predictions of relativistic quantum field theory has been solved with Bell's proposal for ``beables for quantum field theory'' \cite{BellFT}, and that the ongoing search for better pilot wave theories (for example in \cite{DuerrParticleQFT,StruyveEM,Colin,Valentini}) in this domain is motivated by metaphysical rather than empirical problems. Wallace's justification for this claim was even more surprising for me:
\begin{quote}
``\ldots it is debatable whether field-based modificatory strategies will actually succeed in reproducing the predictions of QM. For recall: as I argued \ldots, it is crucial for these strategies that they are compatible with decoherence: that is, that the preferred observable is also decoherence-preferred. \ldots a hidden-variable theory whose hidden-variables are not decoherence-preferred will fail \ldots to recover effective quasiclassical dynamics. And in QFT (at least where fermions are concerned) the pointer-basis states are states of definite particle number, which in general are not diagonal in the field observables.''
\cite{Wallace}
\end{quote}
Instead, I have considered it to be the job of decoherence to explain why we observe particles in a situation where the fundamental pilot wave beables are fields or even something more fundamental. Following Wallace, decoherence is not only unable to do this job, but causes even problems with empirical viability. This was sufficient motivation to evaluate this thesis in more detail.

Once empirical viability is questioned, the natural place to look for is the equivalence proof between pilot wave theory and quantum theory. Considering this proof in section \ref{sec:proof}, we find it necessary to include the state of the observer into the picture. The measurement device is not enough --- once there exists doubts if the observer really observes the state of the measurement device,  one has to consider the full picture. As a consequence we need only the actual, unobserved observer state. The equivalence theorem in this form would remain valid even if the beables themself would be unobservable.

The only loophole which could prevent the recovery of quantum predictions is the question if different macroscopic states can be distinguished by the beables, or, in other words, if these states have negligible overlap as functions of the beables. For the particular case of field beables, it has been shown in \cite{overlap} that the overlaps become negligible for large particle numbers. But field theory is only a particular example. One should expect that more fundamental physical theories give other, more fundamental candidates for the beables --- strings, loops, or cells as proposed in \cite{clm} --- which give the observable fields only in the large distance limit. Their configuration variables (assuming some Lagrange formalism) are natural candidates for pilot wave beables. We obviously have to expect that their connection to the decoherence-preferred observables is different from that in field theory.

This leads to some questions: What is the possible relation between pilot wave beables and decoherence in such a more general situation? While the decoherence-preferred particle variables differ from the field variables, there exist, nonetheless, some nontrivial connection between them. What if there is no such connection? If there is a connection, how does it appear? If it appears in some more or less natural way, what are the consequences? Does this connection have an influence on the ability of pilot wave theories to recover quantum predictions? These are the questions we want to study in this paper.

The answers we find are quite positive for pilot wave theories with more fundamental beables: If there is absolutely no connection between beables and decoherence, we can even prove a theorem that the overlap becomes insignificant. On the other hand, we find that a connection has to be expected: Decoherence starts from a predefined decomposition into systems, and such a decomposition can be obtained in a reasonable way only from the beables. This leads to a non-trivial connection, where different systems have independent sets of beables. We find that this leads to some preference for product states.  While this makes it impossible to apply our general theorem, it appears plausible that in this case the overlap descreases approximately exponentially with an increasing number of system.

\section{The equivalence proof}\label{sec:proof}

Wallace's thesis claims that pilot wave beables have to be decoherence-preferred to recover quantum predictions. To reject it, it is sufficient to consider one equivalence proof which does not depend on this assumption. In our case, a quite standard version of the proof seems sufficient. On the other hand, in another equally standard, but inferior variant Wallace's thesis looks quite natural. The version of the equivalence proof which does not depend on it is enough to falsify Wallace's thesis, but the other variant gives us some interesting hint about the origin of this thesis. A nice summary of the main lesson of our proof has been given by Struyve:
\begin{quote}
``\ldots Saunders expressed a worry for using fields as beables. He expressed some doubts whether localized macroscopic bodies are represented by localized field beables. If not, he claimed, it would be unclear how a pilot-wave model with field beables reproduces the quantum predictions. However, although it is true that a pilot-wave model in which localized macroscopic bodies are represented by localized fields reproduces the quantum predictions, it is by no means a necessary requirement. As long as wavefunctionals corresponding to macroscopically distinct states are non-overlapping in the configuration space of fields, the field beables will display the outcomes of measurements.'' \cite{Struyve},
\end{quote}
If you do not doubt this and are not interested to understand the origin of Wallace's thesis you can skip the remaining part of this section.

Else, let's consider the general scheme of a quantum measurement in pilot wave theory. We restrict ourself in this paper to the simplest case of a measurement with two discrete eigenstates. This seems sufficient: As human beings, we can distinguish only finite numbers of different states, and, given that the number of particles in macroscopic states is much larger than the number of states we can distinguish, the number of the discrete states does not really matter. So assume we have a quantum system $S$ in the initial state $\psi(t_0,q_S)=a_0\psi_0^S(q_S) + a_1\psi_1^S(q_S)$, with $\langle\psi_0^S|\psi_1^S\rangle=0$. The measurement is an interaction with the remaining part of the universe which measures the projection operator $|\psi_1^S\rangle\langle\psi_1^S|$ --- an operator with two eigenvalues, $1$ for $\psi_1$ and $0$ for all other states. The interaction Hamiltonian may be something like $H_{SE}=|\psi_1^S\rangle\langle\psi_1^S|p$, where $p$ is the momentum operator of some unspecified pointer variable. For the system-internal Hamilton operator we simply set $H_S=0$. As the initial state of the wave function of the universe we assume a product state
\begin{equation}
\Psi(t_0,q_{univ}) = \Psi(t_0,q_S,q_{rest}) = \psi^S(t_0,q_S) \psi^{rest}(t_0,q_{rest})
\end{equation} 
where $q_{rest}$ describes the remaining part of the universe. After the interaction, we obtain some superpositional state
\begin{equation}
\Psi(t_1,q_S,q_{rest}) = a_0\psi_0^S(q_S)\psi^{rest}_{0}(t_1,q_{rest}) + a_1\psi_1^S(q_S)\psi^{rest}_{1}(t_1,q_{rest})
\end{equation} 
where $\psi^{rest}_{0}(t_1,q_{rest})$ and $\psi^{rest}_{1}(t_1,q_{rest})$ are already different states.

All the time, the conditional wave function
\begin{equation}
\psi^S(t,q_S) = \Psi(t,q_S,q_{rest}(t))
\end{equation}
obtained by putting the actual value of $q_{rest}(t)$ into the wave function of the universe is well-defined and correctly defines the guiding equation for $q_S(t)$. But in general it does not follow the effective Schr\"{o}dinger equation of the system (which would be trivial for $H_S=0$): During the measurement, it evolves in dependence of $q_{rest}(t)$ and $\Psi(t,q_{univ})$. This is the effective collapse process.

The collapse is finished if the two wave functions $\psi^{rest}_{i}(t_1,q_{rest})$ no longer overlap, and if this property is stable in time. After this, the conditional wave function $\psi^S(t,q_S)$ no longer collapses and becomes an effective wave function, that means, it follows the (in our case trivial) Schr\"{o}dinger equation of the system. If (for fixed $t$) the $\psi^{rest}_{i}(q_{rest})$ do not overlap, the conditions $\psi^{rest}_{i}(q_{rest})\neq 0$ define a decomposition into two sets $Q^i_{rest}$ with
\begin{equation}
q_{rest} \in Q^0_{rest} \cup Q^1_{rest}, \qquad  Q^0_{rest} \cap Q^1_{rest}\cong\emptyset.
\end{equation}
For $q_{rest}\in Q^i_{rest}$ the effective wave function of the system appears to be $\psi_i^S(q_S)$. If the wave function of the universe is in quantum equilibrium, this happens with probability $\abs{a_i}^2$.

All this corresponds nicely with the description given by quantum theory. The only difference is the following: For a quantum measurement being finished, the $q_{rest}$ should be in macroscopically different states. In pilot wave theory, we have to require that the $\psi^{rest}_{i}(q_{rest})$ do not overlap, and that this condition remains stable in time. All we have to prove is that pilot wave theory recovers quantum predictions, thus, we need only one direction. Moreover we can assume that macroscopically different states remain to be macroscopically different in time. Therefore, all we need to recover quantum predictions is the hypothesis that \emph{macroscopically different states $\psi^{rest}_{i}(q_{rest})$ do not overlap as functions of the pilot wave beable $q_{rest}$}.

Note that the configuration $q_{rest}$, as the configuration of the rest of the universe, contains also the configuration of the observer $q_{obs}$ himself. The condition that the observer has seen the result is of course different, more restrictive than the condition that the measurement has been finished. But it is quite close: All we need is that the configurations of the observer $q_{obs}$ are different for different measurement results. Thus, there should be a decomposition into two sets $Q^i_{obs}$ with
\begin{equation}
q_{obs}\in Q^i_{obs} \Rightarrow q_{rest}\in Q^i_{rest}, \qquad  Q^0_{obs} \cap Q^1_{obs}\cong\emptyset,
\end{equation}
in other words, $q_{obs}$ should be already sufficient to distinguish different measurement results.

\subsection{The motivation for Wallace's thesis}

It is also worth to note that the whole consideration is based on really existing objects, not on observed objects. According to pilot wave theory, the wave function of the universe $\Psi(q_{univ})$, as well as the configuration $q_{univ}$, really exist. Thus, the conditional wave function we have constructed based on these ingredients also really exists. As well, the quantum equilibrium distribution is a distribution for the actual, real values of $q_{univ}$. Last but not least, $q_{obs}$ describes the real configuration of the observer, who, therefore, is in really different states if $q_{obs}\in Q^i_{obs}$ for different $i$. None of these objects is observed in some way, nor is there any necessity for them of being observed.

This seems different if we, instead of including the observer $q_{obs}$ into the pilot wave picture, follow the Copenhagen quantum scheme, where the observer is located in the classical part, outside the quantum domain. Everything else seems similar, with $q_{rest}$ being replaced by the configuration of some measurement device $q_m$. But, once we have not included the observer into the picture, we cannot consider the observers real configuration $q_{obs}$. As a consequence, lot's of unnecessary questions appear: The observer has to observe, somehow, the result of the measurement. What does he observe if he looks at the device?  The wave function $\psi^m(q_m)$ of the measurement device or, instead, the actual value of $q_m$? Is it justified to put the actual value $q_m$ into the wave function, or should we use, instead, some observed value of $q_m$? Maybe the observation is fooled like the particle detectors in \cite{Aharonov}? In this incomplete picture, these questions seem unavoidable and difficult to answer. Examples of arguments of this type have been given by Wallace and Brown:
\begin{quote}
``Suppose we accept that it is the entered wavepacket that determines the outcome of the measurement. Is it trivial that the observer will confirm this result when he or she “looks at the apparatus”? No, though one reason fothe effective wave function of the system appears to be $\psi_i^S(q_S)$r the nontriviality of the issue has only become clear relatively recently. The striking discovery in 1992 of the possibility (in principle) of “fooling” a detector in de Broglie-Bohm theory should warn us that it cannot be a mere definitional matter within the theory that the perceived measurement result corresponds to the “outcome” selected by the hidden corpuscles.''\cite{BrownWallace}.
\end{quote}
In the objective picture, the observer is not obliged to ``confirm'' the result.  It is simply an objective fact that for $q_{obs}\in  Q^i_{obs}$ the effective wave function of the system appears to be $\psi_i^S(q_S)$. 
\begin{quote}
``\ldots that the \ldots possibility of fooling detectors casts doubt on the claim by Maudlin \cite{Maudlin} p. 483 that the so-called effective (post-measurement) wavefunction of the object system is \emph{defined} (in part) by the positions of the corpuscles associated with the apparatus'' \cite{BrownWallace}.
\end{quote}
It is exactly reverse: Because the conditional wave function is \emph{defined} by $q_{rest}$, and not by any results of observations, the possibility of fooling detectors does not matter at all.

It is this incomplete consideration of the measurement process which can be easily used to motivate Wallace's thesis: The incomplete consideration can be interpreted as a reduction of general measurements to measurements of the beables $q_m$. How does this fit with decoherence?  To reduce a decoherence-preferred measurement to a non-decoherence-preferred one seems awkward, in conflict with decoherence. Thus, it seems necessary to require that the measurements of beables have to be decoherence-preferred. But once we do not have to consider any measurement of $q_m$, this argument fails.

\subsection{The classial limit}

Nonetheless, let's consider in more detail if anything can go wrong with the classical limit. Last but not least, Wallace focusses in his consideration of the quasiclassical domain:
\begin{quote}
``\ldots a hidden-variable theory whose hidden-variables are not deco\-herence-preferred will fail \ldots to recover effective quasiclassical dynamics.''
\cite{Wallace}
\end{quote}
And the classial limit works in pilot wave theory in a very different way than in usual quantum mechanics, a way which is centered around the beables. Indeed, assume we have a wave function with $\abs{\Psi}=\mathrm{const}$. Looking at the quantum potential, we find it being zero. As a consequence, for the phase $S(q,t)$ and the trajectory $q(t)$ the classical Hamilton-Jacobi equation holds. Thus, $q(t)$ and $p(t)=\nabla S(q,t)$ is a solution of the classical Hamilton equations, and we are in a completely classical situation. Once all our configuration is defined by $q(t)$, we cannot distinguish our world from a purely classial world, and, in particular, cannot observe nor $S(q,t)$ nor the wave function.  This defines an essential difference between pilot wave theory and usual quantum theory: In usual quantum theory, a state with $\abs{\Psi}=\mathrm{const}$ is a widely distributed, non-localized quasiclassical state, very different from a wave-packet localized in $p$ and $q$ we would need for the classical limit in pure quantum theory.

Is there any conflict between this pilot-wave-specific picture of the classical limit and the one provided by decoherence? The limit is described in terms of the beable variables, not in terms of decoherence-preferred observables, so, at least in principle, one could expect some problems. But I see no base for any conflict, because in this limit we do no longer have to choose between incompatible measurements: In the classical limit, all observables $f(p,q)$ commute, that means, in particular, that a measurement of the decoherence-preferred observables, whatever they are, is no longer incompatible with other observations.

\subsection{The remaining problem: The overlap of macroscopic states}\label{sec:overlapDefinition}

As we have seen, in the more complete pilot wave picture, which includes the observer as well, we do not have to consider observations of beables. So we avoid all problems related with unobservable trajectories or fooled detectors. The picture is completely based on objective entities as $q_{rest}$. We obtain an objective collapse picture described by the conditional wave function $\psi^S(t,q_S)$, which with the correct Born probability \emph{actually becomes} one of the eigenstates $\psi^S_i(q_S)$. What remains to check is only one assumption: That \emph{macroscopically different states do not overlap.}

Of course, not overlapping at all is only an idealization. What we need is that the overlap is not significant. To evaluate the significance of overlaps, we need a precise definition. We have proposed and used such a definition in \cite{overlap}: The overlap has been defined there as a relation $\rho(\psi_0|\psi_1)$ between two wave functions $\psi_0(q)$ and $\psi_1(q)$:
\begin{equation}\label{overlapdef}
\rho(\psi_0|\psi_1) = \int \d q\ \chi_{|\psi_0(q)|<|\psi_1(q)|} |\psi_0(q)|^2.
\end{equation}
This definition has a simple interpretation as the probability that a particle guided by $\psi_0(q)$ appears to be in a region where $\psi_1(q)$ dominates. If all relevant overlaps are sufficiently small compared with $1$, one can replace the $\psi_i(q)$ by non-overlapping approximations so that the probability that this approximation becomes relevant can be expressed in terms of the overlaps, and, therefore, appears to be sufficiently small.

\section{What happens if there is no connection between beables and decoherence?}\label{sec:noConnection}

If the overlap between macroscopically different states becomes insignificant or not depends on the particular pilot wave theory --- on it's Hamilton operator as well as on it's choice of the beables. Decoherence plays an important role in the process of macroscopic amplification, therefore it can have a nontrivial influence on the overlap.

Despite this it appears possible to prove a quite general theorem about the overlap between two states. This theorem assumes that there is no connection at all between decoherence and the pilot wave beables: The result of the decoherence-governed amplification process are simply two different states, and whatever the connection between them from point of view of the decoherence process, from point of view of the pilot wave beables they have no connection at all. Thus, they have to be simply two independent random states in a sufficiently large-dimensional Hilbert space. In this case, we can compute the expectation value of the overlap and consider it's dependence on the dimension of the Hilbert space. It appears that in this limit the overlap becomes insignificant.

Such a general theorem is important even for those cases where we cannot apply it --- it allows to play the scientific game known as ``burden tennis'': What happens if there is no connection at all is a reasonable ``null hypothesis''. If one thinks that in a particular case the situation is different, one has the burden of the argumentation. To show only that the general theorem is not immediately applicable seems not sufficient, it has to be at least plausible that the resulting overlaps will become significant.

\begin{figure}
\centerline{\includegraphics[angle=0,width=0.9\textwidth]{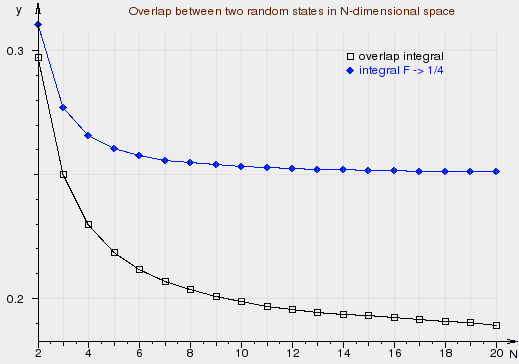}}
\caption{\label{fig:points}
Monte Carlo simulation results for overlap between two random states in $N$-dimensional space. The slowing down of the initial decrease rate is explained by the integral $F$ of \eqref{F}, which fastly approaches it's limit $\frac{1}{4}$.}
\end{figure}

For simplicity, we start with the consideration of a real N-dimensional Hilbert space. We want to compute the following expectation value $E$ for the overlap between two states $\psi_0$, $\psi_1$, which are independent and randomly distributed on the unit sphere $S^{N-1}$:

\begin{equation}\label{expectation}
E = \int\limits_{S^{N-1}} \frac{\d\Omega_0}{\vol{S^{N-1}}} \int\limits_{S^{N-1}} \frac{\d\Omega_1}{\vol{S^{N-1}}} \sum\limits_{i=1}^N \chi_{\abs{\psi^i_0}<\abs{\psi^i_1}} \abs{\psi^i_0}^2.
\end{equation}

For $N=2$, this integral can be taken explicitly and gives $\frac12 - \frac{2}{\pi^2}\approx 0.297357$. For small values of $N$ one can use Monte Carlo simulation to compute the integral. The result, as presented in figure \ref{fig:points}, does not look promising: After an initial decrease down to $\approx 0.2$, the rate of decrease dramatically slows down, suggesting the possibility of a nontrivial lower bound (see ). A cross-check with another way to compute the integral has given an even worse result: An approach to a fixed limit $\frac{1}{4}$. This was caused by a program error, but a remarkable and helpful one: It has shown that the integral $F$ of equation \eqref{F} below behaves very well for large $N$. This observation was a good starting point to find the proof below.

\begin{theorem}  \label{th:real}
For $N\to\infty$ the limit of the expectation value for the overlap $E$ in \eqref{expectation} is zero.
\end{theorem}

\begin{proof} To estimate the integrals of type $\int_{S^{N-1}} \d\Omega f(\psi)$ we replace them by integrals over the unit cube $\loo{\psi}\le 1$ for the function $f(\psi/\ltwo{\psi})$. This requires the introduction of a weight factor. The weight is the inverse of the volume of the infinitesimal conus $\frac{1}{N}\loo{\frac{\psi}{\ltwo{\psi}}}^{-N}$, which becomes projected on a given surface element $d\Omega$. Thus, we obtain the following rule for replacement of the integrals:
\begin{equation}\label{rule}
\int_{S^{N-1}}\frac{\d\Omega_\psi}{\vol{S^{N-1}}} f(\psi)  = \int_{-1}^{1}\frac{\d\psi_0}{2}\cdots\int_{-1}^{1}\frac{\d\psi_N}{2} f(\frac{\psi}{\ltwo{\psi}}) N\frac{\loo{\psi}^N}{\ltwo{\psi}^N}.
\end{equation}
Given the form of the weigth factor, it seems useful to split the integral $E$ into two parts --- a ``localized part'' $E_{local}$ containing states with $\loo{\psi_k}\ge (1-\varepsilon)\ltwo{\psi_k}$ for above states $\psi_k$, and a remaining part $E_{rest}$ containing everything else. Let's at first evaluate the remaining part. Here, in one of the two integrals we can use the estimate $\loo{\psi_k}<(1-\varepsilon)\ltwo{\psi_k}$, while for the other part the estimate $\loo{\psi_k}<\ltwo{\psi_k}$ holds (for every $\abs{\psi^i}$, and, therefore, for the maximum as well, we have $\abs{\psi^i}\le\ltwo{\psi}$). This gives
\begin{equation}
E_{rest} \le N^2(1-\varepsilon)^N F
\end{equation}
with
\begin{equation}\label{F}
F =
\int\limits_{-1}^{1}\frac{d\psi^1_0}{2}\cdots\int\limits_{-1}^{1}\frac{\d\psi^N_0}{2}
\int\limits_{-1}^{1}\frac{d\psi^1_1}{2}\cdots\int\limits_{-1}^{1}\frac{\d\psi^N_1}{2}
\sum\limits_i \chi_{\abs{\psi^i_0}<\abs{\psi^i_1}} \frac{\abs{\psi^i_0}^2}{\ltwo{\psi_0}^2}.
\end{equation}
But because $\abs{\chi}\le 1$, $\sum_i \abs{\psi^i_0}^2 = \ltwo{\psi_0}^2$, and the remaining integrals are probability measures on the unit cube, we immediately obtain $F\le 1$. In fact, it can be proven as well that this integral has the limit $\frac{1}{4}$. It fastly approaches this limit, as can be seen in figure \ref{fig:points}. But $F\le 1$ is sufficient to give
\begin{equation}
E_{rest} \le N^2(1-\varepsilon)^N \to 0  \qquad\text{for}\qquad N\to\infty.
\end{equation}
It remains to estimate the localized part $E_{local}$. The reason for naming it that way is that the states in this part are, for small $\varepsilon$, sharply localized around their maximal value $\loo{\psi}=\abs{\psi^i}$. Indeed, it follows from $\abs{\psi^i} \ge (1-\varepsilon)\ltwo{\psi}$ that
\begin{equation}\label{localOverlap}
\sum_{j\neq i} \abs{\psi^j}^2 \le 2\varepsilon\ltwo{\psi}^2.
\end{equation} 
Now, the whole integral $E_{local}$ splits into $N^2$ parts $E_{ij}$ defined by the maximal coordinates of the two wave functions $\loo{\psi_0}=\abs{\psi^i_0}$, $\loo{\psi_1}=\abs{\psi^j_1}$. The function $\chi_{\abs{\psi^i_0}<\abs{\psi^i_1}} \abs{\psi^i_0}^2$ we have to integrate in \eqref{expectation} can be estimated for the diagonal parts $E_{ii}$ by $1$, and, because of \eqref{localOverlap}, for the non-diagonal parts $E_{ij}, i\neq j$, by $2\varepsilon$. If we, instead, replace this function with $1$, we obtain an integral over a probability probability measure on some subset, which gives something $\le 1$. In this integral, all parts $E_{ij}$ appear in a symmetric way, thus, the integrals over each of them should not depend on $i$ and $j$. Thus, each part $E_{ij}$ can give at most $N^{-2}$. This gives the following estimate for the local part of the integral:
\begin{equation}
 E_{local} \le \frac{1}{N} + 2 \varepsilon.
\end{equation} 
Thus, for every given $\varepsilon'>0$ we can choose $\varepsilon=\frac{1}{3}\varepsilon'$ so that for sufficiently large $N$
\begin{equation}
E \le \frac{1}{N} + 2 \varepsilon + N^2(1-\varepsilon)^N < \varepsilon'.	\qedhere
\end{equation}
\end{proof} 

It seems worth to note that in this proof we have not even required that the states $\psi_0$, $\psi_1$ should be orthogonal.

What changes if we consider, instead, the usual complex Hilbert space? Not much. We have to double the dimension of space. The invariant measure on $\mathbb{CP}^N$ can be obtained as the image of the invariant measure on $S^{2N-1}$ from the standard Hopf projection $S^{2N-1}\to \mathbb{CP}^N$, so that we can simply use the same measure $\d\Omega$ on $S^{2N-1}$ as the probability measure. The function we have to integrate changes only in a minor way: instead of one-dimensional condition $\abs{\psi^i_0}<\abs{\psi^i_1}$, we have to use the complex analogon $\bar{\psi}^i_0\psi^i_0 < \bar{\psi}^i_1\psi^i_1$, and to multiply it with $\bar{\psi}^i_0\psi^i_0$:
\begin{equation}\label{expComplex}
E_c = \int\limits_{S^{2N-1}} \frac{\d\Omega_0}{\vol{S^{2N-1}}} \int\limits_{S^{2N-1}} \frac{\d\Omega_1}{\vol{S^{2N-1}}} \sum\limits_{i=1}^N \chi_{\bar{\psi}^i_0\psi^i_0 < \bar{\psi}^i_1\psi^i_1} \bar{\psi}^i_0\psi^i_0.
\end{equation}
For this integral, the following theorem holds:
\begin{theorem} \label{th:complex}
For $N\to\infty$ the limit of the expectation value for the overlap $E_c$ in \eqref{expComplex} is zero.
\end{theorem}
\begin{proof} The only place which requires modification in comparison with the previous proof is that we obtain an overlap of order $1$ for $4N$ of the $4N^2$ localized pairs $E_{ij}$. Everything else remains unchanged. Thus, we obtain
\begin{equation}
E_c \le \frac{1}{N} + 2 \varepsilon + (2N)^2(1-\varepsilon)^{2N} < \varepsilon'
\end{equation}
with the same conclusion.
\end{proof}

\section{The counterexample: One-particle theory in field ontology} \label{sec:counterexample}

Thus, if there is no connection at all, everything is fine. But this is an assumption we cannot rely on. At the other extreme, we have an explicit counterexample where the overlap never becomes negligible, so that the standard equivalence proof fails. This example is one-particle theory, artificially described with a field ontology. In this case, in agreement with the results of Struyve \cite{Struyve}, we have found in \cite{overlap} that all states have an overlap of at least $0.18169(\pm 1)$ --- the universal overlap for orthogonal one-particle states in field theory.

What are the characteristic properties of this example which lead to this failure of the equivalence proof? We can find several very special properties, which, in combination, allow to explain this failure:

\begin{itemize}

\item First, there is an additional non-trivial conservation law --- the conservation of particle number. Without this, the counterexample would not work. It would be possible to obtain states with higher number of particles, which already have negligible overlap. In particular, an approximate conservation of particle number is not sufficient: It only increases the decoherence time.

\item Second, we have an extremal choice of the initial value --- only one particle. With a much larger number of particles, say, $10^80$, the problem would disappear as well.

\item Last but not least, we have a special relation between the field ontology and the decoherence-preferred observables: The decoherence-preferred states with fixed particle number are localized in a certain environment of the vacuum state in the field variables. As well, without this special relation, the problem would disappear, as a consequence of theorem \ref{th:complex}.

\end{itemize}

Given the thesis of Wallace that the pilot wave variables have to be decoherence-preferred, there is a certain irony in the observation that some non-trivial connection between pilot wave variables and decoherence-preferred variables is necessary for the things going wrong.

\section{Is there a connection between decoherence-preferred observables and pilot wave beables?}\label{sec:connection}

Unfortunately, in the physically interesting cases, in particular in field theory, we cannot apply our theorems \ref{th:real}, \ref{th:complex}, because we have some nontrivial connection between decoherence-preferred observables and beables. Even if, in the case of field theory, there is no identity, they are far away from being completely unrelated, so that one cannot assume that the macroscopic states generated by decoherence are independent random states.

The question we want to consider now is if there is some natural connection between them. 

In this context it becomes important to observe that the Hamilton operator taken alone does not allow to define ``the'' decoherence-preferred basis \cite{kdv, kdv2}. One needs more, namely a ``decomposition into systems'', to define a decoherence-preferred basis. Where does this ``decomposition into systems'' come from? Decoherence theory remains silent, it takes the decomposition as given. The particular subdivision into the various systems we see around us --- stars, planets, stones, cats, and human beings --- can be defined only in some environment of the actual state of the universe. For other possible states of the universe --- for example, states where the Earth does not even exist --- the system which defines, for example, a given cat simply does not make sense. Thus, in interpretations like many worlds, which do not have some well-defined actual configuration of the universe as pilot wave interpretations, we have also no natural subdivision around this state into systems.

Leaving this problem to these interpretations (see \cite{kdv, kdv2} for more) let's continue with the situation in pilot wave theory. Here we have some basis for a derivation of a decomposition into subsystems we need for decoherence. Indeed, pilot wave theory defines, at every moment of time $t$, some well-defined configuration of the universe $q_0 = q(t) \in Q$. This configuration can be used to develop some linear theory in the tangent space $\T \cong TQ|_{q(t)}$ of the configuration space $Q$ at $q(t)$. The various systems we see around us may be identified with linear subspaces of $\T$. This is, of course, not the place to consider possible details of such a construction. The only point we want to make here is that the tangent space $\T$ at $q(t)$ is a very natural ecological niche for the definition of all these subsystems.

If we use this place to define the subsystems, the decomposition $\T \cong \prod \T_{S^i}$ of the tangent space gives a corresponding decomposition $\L(\T,\C)\cong \bigotimes \L(\T_{S^i},\C)$ of the Hilbert space $\L(\T,\C)$, which approximates $\H\cong \L(Q,\C)$ in the environment of $q_0$. The notion of ``approximation'' used here is also a specific property of pilot wave theories: The guiding equation for some finite time interval depends only on the wave function in the relevant environment of $Q$, so that the restriction of the global wave function to some environment can be considered as a meaningful approximation --- an argumentation which is not valid in other interpretations. Now, the decomposition $\L(\T,\C)\cong \bigotimes \L(\T_{S^i},\C)$ already defines a natural decomposition of $\L(\T,\C)$ into systems which can be used to start decoherence considerations.

If we use a construction of this type to obtain the decomposition into systems, we obtain, automatically, a natural connection between the pilot wave beables, which describe the configuration space $Q$, and the decoherence-preferred observables derived starting with the decomposition $\T \cong \prod \T_{S^i}$: For each of the subsystems $\T_{S^i}$ we have, by construction, separate pilot wave beables $q_{S^i}=(q^k_{S^i})$, that means, the decomposition into systems coinsides with a decomposition of the pilot wave beables. 

From point of view of pure theory, this relation is problematic: Our theorem  \ref{th:complex}, which would be sufficient for our purpose, is not applicable once something very special happens --- we cannot rely on the assumption that the states created by decoherence are independent random states.

\section{Overlaps between product states}\label{sec:productStates}

Is there a special property of the states which could prevent us from applying theorem \ref{th:complex}? Such a property should be stable in time --- else, it could be rejected as unimportant. But decompositions into subsystems are useful exactly because the subsystems develop almost independently in time. In particular states which are product states initially remain product states with high probability.

Moreover, local interactions of some systems with their local environments tend to destroy superpositions between such product states. In pilot wave theories, this works in the following way: If a system $|\psi\rangle$ in a superpositional state $|\psi_1\rangle|\phi_1\rangle+|\psi_2\rangle|\phi_2\rangle$ with some far away system $|\phi\rangle$ interacts with its local environment $|\theta\rangle$, this gives some superposition
\begin{equation}\label{eq:superposition}
|\psi_1\rangle|\phi_1\rangle|\theta_1\rangle+|\psi_2\rangle|\phi_2\rangle|\theta_2\rangle.
\end{equation}
Now, we can use the actual value of the pilot wave variable $q^\theta$ of the local environment $|\theta\rangle$ to obtain an effective wave function for the two systems. If $q^\theta$ is not in the overlap of the two functions $\theta_i(q^\theta)$, this gives or $|\psi_1\rangle|\phi_1\rangle$, or $|\psi_2\rangle|\phi_2\rangle$, destroying the superposition between the pure product states. 

Giving these two effects --- stability of product states because of independent evolution of the systems, and reduction to product states because of independent interactions of particular systems with their particular environments --- we have to expect that the states obtained by decoherence have a tendency to be product states. Given the connection between decoherence and pilot wave theory, products of states of subsystems appear to be product states in terms of the corresponding pilot wave beables $q_{S^i}$:
\begin{equation}
\psi_a(q)=\prod \psi_{ai}(q_{S^i}), \qquad a\in\{1,2\}.
\end{equation}
Thus, the theorem \ref{th:complex} is not sufficient in a situation where the systems used by decoherence are subsystems of beables, because the probability distribution of the states created by decoherence is not random, but has a preference for product states. We cannot present a comparable proof for this general situation. Nonetheless, it appears sufficiently plausible that nothing dangerous happens.

The situation is similar to the case of $n$ orthogonal particles we have considered in \cite{overlap}: The n-particle states considered there have been also product states of one-particle states. As in this case, we have no simple product rule --- the product of the overlaps gives only a lower bound of the overlap of the product states. Nonetheless, there are good reasons to expect that the decrease is sufficiently fast: Indeed, it seems not unreasonable to expect that the overlap of the product states decreases in a similar way as the overlap of the n-particle states. In this case, the decrease would be approximately exponential in the number of different degrees systems $q_{S^i}$, which is much more than we need to distinguish macroscopically different states.

The plausibility arguments we have used in \cite{overlap} for the n-particle product state can be applied in the case of general product states: The distance between the maximal values of $\abs{\psi_a(q)}$ of the two product functions $\psi_a(q)=\prod \psi_{ai}(q_{S^i})$ can be easily computed. Indeed, whatever the functions $\psi_{ai}(q_{S^i})$ which participate in products $\abs{\psi_i(q)}$ of positive-valued functions $\abs{\psi_{ik}(q_{S^i})}$, the maximum will be reached in the points where every function reaches it's maximum. Therefore, the distance $\Delta$ between the maxima of the products can be obtained from the corresponding distances for the factors $\Delta_i$ by the simple formula $\Delta^2 = \sum \Delta_i^2$. Thus, it is increasing with the number of factors.

Some plausible expectations about the behaviour of the overlap can be obtained if we consider the case of equal factor functions $\psi_{ai}(q_{S^i})=\widetilde{\psi}_{a}(q_{S^i})$ for all $i$. In this case, the distance between the maxima increases as $\sqrt{n}$. Let's also look at the line connecting the maxima. Assume, the two functions $\widetilde{\psi}_{a}(q_{S^i})$  have only one maximum $\psi^a_{max}$ at $q^a_{max}$ and decrease with distance from the maximum. Then, in the one-dimensional case, there can be only one local maximum $q_{max}\in [q^1_{max},q^2_{max}]$ of the overlap function --- the point $q_{max}$ defined by $\widetilde{\psi}_{1}(q_{max})=\widetilde{\psi}_{2}(q_{max})=\psi_{max}$. Then, on the line connecting the two maxima in the n-dimensional case, the overlap function reaches it's maximum at the point $(q_{max},\ldots,q_{max})$, and this maximal value is $\psi_{max}^n$. Now, given that the integrals $\ltwo{\widetilde{\psi}_{a}}$ of functions with maxima $(\psi^a_{max})^n$ give $1$, we can expect the order of the corresponding overlap integral as being roughly proportional to the relation of the maxima. This relation is $(\psi_{max}/\psi^a_{max})^n$, thus, it decreases exponentially with $n$. This gives a plausibility argument that the overlap integral will decrease approximately exponentially with the number of systems.

An extremal case where we can compute the overlap exactly is the case where the wave functions can have only two values: $0$ and $1$. In these cases, the overlap for each system consists of the normalized area of the configuration space where above functions are $1$. For product states of this type, the overlap is the product of the normalized areas of the subsystems, thus, the product of the probabilities of overlap for the subsystems. If all subsystems are of the same type, thus, give the same overlap proability $p$, we obtain $p^N$ as the overlap probability for $N$ systems. This gives, again, the expectation that the overlap decreases exponentially with the number of systems.

In \cite{overlap} we have made  some numerical computations of overlap integrals for product states for some other functions. The results also support the thesis that the overlap decreases approximately exponentially with the number of systems.

All these are, of course, only plausibility arguments. As well, we have considered here only the two extremal situations: independent random states, and pure product states.  The general situation will be some mixture.

Nonetheless, with these plausibility arguments the advantage in burden tennis is already quite large: We have not only the general theorem as the default assumption. We have also found lot's of arguments for the most plausible derivation from the default assumption --- product states instead of random states. Additionally there is a large gap between the fast, exponential decrease which seems plausible and the minimal decrease which would be sufficient to preserve the phenomena, given the large numbers of particles in macroscopic states.

Given that the apparatus of quantum theory depends on such uncertain notions as macroscopic states, we should not expect much more. Thus, the current situation seems close enough to the ideal one. Designing new pilot wave theories, we are free to introduce whatever beables we prefer. While there is some theoretical possibility that the equivalence theorem fails for the resulting theory, it is sufficiently plausible that in general it does not.

\section{Conclusions}

Motivated by Wallace's thesis that pilot wave beables have to be decoherence-preferred to allow a recovery of quantum theory predictions, we have studied the connection between pilot wave beables and decoherence.

Wallace's thesis has to be rejected: On the contrary, if we have no connection at all between decoherence and pilot wave beables, we can prove that the overlaps between macroscopic states become insignificant, which is sufficient to recover quantum predictions.

On the other hand, we have found that we have to expect a nontrivial connection between decoherence and pilot wave beables. The point is that decoherence depends on the choice of a decomposition into systems, and we have found that the natural environmental niche for such a decomposition is closely connected with the beables. In this case, we have to expect that the states created by decoherence are not random, but have some preference for product states. Fortunately, product states do not seem problematic as well. We have presented plausibility arguments that the overlap between product states decreases exponentially with the number of systems. That means, in above extremal cases --- random states as well as pure product states --- the overlap decreases sufficiently fast. Thus, there is no reason to expect something wrong happens in the general case.

In special circumstances, the overlap may always remain significant: In the counterexample of one-particle theory described with field beables, we have an additional conservation law (particle number) and a very special initial value (one) for the conserved value as particular circumstances which force the overlap to remain significant. In the general case, it seems quite implausible that such things happen. Thus, given the results of this paper, one can leave to burden of argumentation to those who doubt that a particular pilot wave theory is viable.

The rejection of Wallace's thesis allows to use decoherence as a useful, important tool in pilot wave theories. The tangential space \T\/ at the actual configuration $q(t)$ of the universe gives an ecological niche for decompositions of the degrees of freedom of the universe into systems. Starting with such a decomposition into systems, we can apply decoherence to find out which are the most stable, most accessible observables. These observables are not necessarily the beables themself --- instead, as in the example of field theories, some sort of particles or quasiparticles may be decoherence-preferred, while the fundamental beables are not. In such a context, decoherence appears to be a useful tool to find out what will be observable.

Thus, the connection between decoherence and pilot wave theory is symbiotic: The actual configuration $q(t)$ and its environment defines a place for the decomposition into systems, which is a necessary prerequisite for decoherence. In return, decoherence explains what we observe, even if we, instead of the fundamental beables, observe something different like some sorts of quasiparticles.

\end{document}